\documentclass[runningheads]{llncs}
\usepackage{graphicx}
\usepackage{amsmath}
\usepackage{algorithm}
\usepackage[noend]{algpseudocode}
\usepackage{tikz}
\begin{document}
\title{Real-Time Systems Modeling and Analysis\thanks{University of Colorado Boulder}}
%
%
\author{Lakhan Shiva Kamireddy\orcidID{0000-0001-6007-5408}}
\authorrunning{LS. Kamireddy}
%
\institute{University of Colorado, Boulder CO 80302, USA\\
\email{\{lakhan.kamireddy\}@colorado.edu}}
\maketitle              
\begin{abstract}
This paper is a survey of extensions to finite automata theory to model real-time systems as well as systems exhibiting mixed discrete-continuous behavior. Real-time systems maintain a continuous and timely interaction with the environment, often adhering to some timing constraints. Therefore, the finite automata theory is extended to measure real-time values and accept or reject runs on a class of automata known as timed-automata, upon satisfying some timed properties. The automata modeling the mixed discrete-continuous behavior of hybrid systems has its continuous-time dynamics described using ordinary differential equations for the state space and discrete-time dynamics describing the control decisions. Based on these dynamical system models, we likewise extend the finite automata theory to describe the behavior of hybrid systems using Hybrid Automata. We further study some applications of this class of systems, sometimes referred to as Cyber-physical systems and perform a case-study on Peterson's Mutual Exclusion protocol using Uppaal.

\end{abstract}
\section{Introduction}
\subsection{Real-time systems and automata based modeling}
Real-time systems maintain timely interaction with the environment. The timing constraints are crucial to such systems and it may lead to dramatic consequences when these constraints are not met. We therefore would like to model such systems using automata theory and model check the timed properties of real-time systems. However, we don't have the necessary resources to model time in finite automata. In 1994, Alur and Dill published the results of their study as the theory of timed automata [1], thereby solving this problem. Timed automata theory is an extension to finite automata theory that describes a way of measuring real-time event occurrences in such systems. It equips us with the required resources in verifying these real-time systems and obtaining provably correct guarantees of timed properties in them. We will study various properties of timed automata and prove some of them in section II. We will also look at a class of systems known as Hybrid Systems that have mixed discrete-continuous behavior. We will look at the hybrid automata theory for modeling this behavior in section III. In section IV we will look at some of the applications of systems known as Cyber-Physical systems that have these hybrid properties. In section V we look at a case-study analysis of the peterson's mutual exclusion protocol in a tool for modeling and verifying real-time systems, and discuss open problems and conclusions in section VI.

\subsubsection{Real-time systems} Real-time systems are encountered in many instances, and we interact with them more often than we realize. Some examples of real-time systems are event response systems like airbag systems, closed-loop control like cruise control system in a car, aircraft control systems, cardiac pacemakers. Before looking at automata for modeling real-time systems, let us first take a look at a class of words that have an embedded time component within them.

\subsubsection{Timed words} An alphabet $\Sigma$ is defined over a finite set of letters. A timed word is a tuple (w, t) where w is a word over the alphabet $\Sigma$, w=$a_1$$a_2$$a_3$...$a_n$ and t is a time sequence, t=$t_1$$t_2$$t_3$...$t_n$ where $t_1{\leq}t_2{\leq}t_3{\leq}$...$t_n$ and $t_i$ $\in$ $R_{\geq0}$, the set of non-negative real numbers. A timed language is a set of words in L such that $L$ $\in$ $T{\Sigma}^*$ is a property over timed words [2].

\subsubsection{Automaton modeling timed properties} Fig. 1. shows an automaton modeling a lamp. It also describes the timed properties of the lamp. The lamp starts in off state. When the switch is pressed in off state, it takes the transition to low state while resetting a clock, y to zero. From low state if the switch is pressed again before 5 seconds have elapsed, it transitions to bright state. From low state if the switch is pressed after 5 seconds have elapsed, it takes a transition back to the off state and accepts. From bright state, if the switch is pressed, it transitions to off state and accepts.

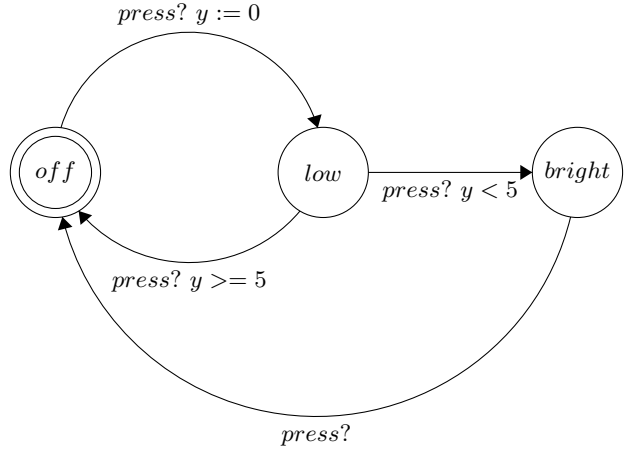
\begin{figure}
\begin{center}
\begin{tikzpicture}[scale=0.2]
\tikzstyle{every node}+=[inner sep=0pt]
\draw [black] (22.1,-26.6) circle (3);
\draw (22.1,-26.6) node {$off$};
\draw [black] (22.1,-26.6) circle (2.4);
\draw [black] (39.9,-26.6) circle (3);
\draw (39.9,-26.6) node {$low$};
\draw [black] (56.8,-26.6) circle (3);
\draw (56.8,-26.6) node {$bright$};
\draw [black] (22.463,-23.636) arc (-196.62661:-343.37339:8.91);
\fill [black] (39.54,-23.64) -- (39.79,-22.73) -- (38.83,-23.01);
\draw (31,-16.78) node [above] {$press?\mbox{ }y:=0$};
\draw [black] (38.36,-29.161) arc (-39.96724:-140.03276:9.604);
\fill [black] (23.64,-29.16) -- (23.77,-30.09) -- (24.54,-29.45);
\draw (31,-33.1) node [below] {$press?\mbox{ }y>=5$};
\draw [black] (56.344,-29.561) arc (-13.69418:-166.30582:17.388);
\fill [black] (22.56,-29.56) -- (22.26,-30.46) -- (23.23,-30.22);
\draw (39.45,-43.33) node [below] {$press?$};
\draw [black] (42.9,-26.6) -- (53.8,-26.6);
\fill [black] (53.8,-26.6) -- (53,-26.1) -- (53,-27.1);
\draw (48.35,-27.1) node [below] {$press?\mbox{ }y<5$};
\end{tikzpicture}
\caption{Automaton modeling timed properties of a lamp} \label{fig1}
\end{center}
\end{figure}

\subsection{Timed Automata} A timed automaton is composed of a finite automaton and a finite set of real-valued clocks. All the clocks values increase at the same rate. Guards can be placed on the transitions of the automaton using which we can enable or disable that transition, thereby constraining the behavior of the automaton and describing timed properties through the language of the automaton. Guards are comparisons of clock values with constants that are non-negative rational numbers and can either evaluate to true or false, g: $x{\leq}c$ $\mid$ $x{\geq}c$ $\mid$ ${\neg}g$ $\mid$ $g{\land}g$ where $x$ $\in$ $Clocks$, $c$ $\in$ $Q_{\geq0}$. The clocks can be reset.

\subsubsection{Formal defn.} A timed automaton M is defined as a tuple, M = (Q, $\Sigma$, C, Inv, $\delta$, $q_0$, F). Q is the finite set of states. $\Sigma$ is a finite set of actions. C is a finite set of clocks. Inv associates each location with an invariant. $\delta$ is a set of transitions. (q,a,g,r,q') is a transition from \textit{q} to \textit{q'}, executing an action \textit{a}, satisfying a guard \textit{g} and clock resets in \textit{r}. There are two types of transitions namely a location switch and a time switch. By elapsing time, and satisfying the location invariant, the automaton can stay in the same state thus making a time switch. By satisfying the guard and taking a transition to another location, the automaton makes a location switch.

\section{Timed Automata properties}
\subsection{Non-deterministic timed automata}
There can be non-determinism concerning location and non-determinism concerning time. $\epsilon$ transitions may also introduce non-determinism w.r.t location. Fig. 2. shows an example of location non-determinism. Upon reading a \textit{b} from the state $s_1$ we cannot determine deterministically if the automaton is going to stay in the state $s_1$ or if it is going to transition to the state $s_2$. Fig. 3. shows an example of time non-determinism. Upon reading an \textit{a} from the state $s_1$ we cannot determine deterministically the time of occurrence of the event (reading an a) as long as it satisfies the guard condition t $\leq$ 4 and location invariant t $\leq$ 5.

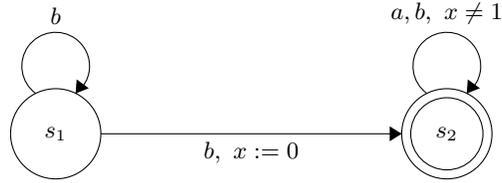
\begin{figure}
\begin{center}
\begin{tikzpicture}[scale=0.2]
\tikzstyle{every node}+=[inner sep=0pt]
\draw [black] (28.5,-25.4) circle (3);
\draw (28.5,-25.4) node {$s_1$};
\draw [black] (54.5,-25.4) circle (3);
\draw (54.5,-25.4) node {$s_2$};
\draw [black] (54.5,-25.4) circle (2.4);
\draw [black] (31.5,-25.4) -- (51.5,-25.4);
\fill [black] (51.5,-25.4) -- (50.7,-24.9) -- (50.7,-25.9);
\draw (41.5,-25.9) node [below] {$b,\mbox{ }x:=0$};
\draw [black] (27.177,-22.72) arc (234:-54:2.25);
\draw (28.5,-18.15) node [above] {$b$};
\fill [black] (29.82,-22.72) -- (30.7,-22.37) -- (29.89,-21.78);
\draw [black] (53.177,-22.72) arc (234:-54:2.25);
\draw (54.5,-18.15) node [above] {$a,b,\mbox{ }x\neq1$};
\fill [black] (55.82,-22.72) -- (56.7,-22.37) -- (55.89,-21.78);
\end{tikzpicture}
\caption{Non-deterministic Timed Automaton (location non-determinism)} \label{fig2}
\end{center}
\end{figure}

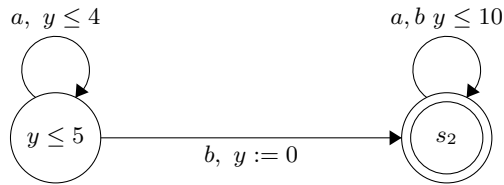
\begin{figure}
\begin{center}
\begin{tikzpicture}[scale=0.2]
\tikzstyle{every node}+=[inner sep=0pt]
\draw [black] (28.5,-25.4) circle (3);
\draw (28.5,-25.4) node {$y\leq5$};
\draw [black] (54.5,-25.4) circle (3);
\draw (54.5,-25.4) node {$s_2$};
\draw [black] (54.5,-25.4) circle (2.4);
\draw [black] (31.5,-25.4) -- (51.5,-25.4);
\fill [black] (51.5,-25.4) -- (50.7,-24.9) -- (50.7,-25.9);
\draw (41.5,-25.9) node [below] {$b,\mbox{ }y:=0$};
\draw [black] (27.177,-22.72) arc (234:-54:2.25);
\draw (28.5,-18.15) node [above] {$a,\mbox{ }y\leq4$};
\fill [black] (29.82,-22.72) -- (30.7,-22.37) -- (29.89,-21.78);
\draw [black] (53.177,-22.72) arc (234:-54:2.25);
\draw (54.5,-18.15) node [above] {$a,b\mbox{ }y\leq10$};
\fill [black] (55.82,-22.72) -- (56.7,-22.37) -- (55.89,-21.78);
\end{tikzpicture}
\caption{Non-deterministic Timed Automaton (time non-determinism)} \label{fig3}
\end{center}
\end{figure}

\subsection{Deterministic timed automata}
A timed automaton is deterministic if a) It has only one initial location, b) It doesn't have $\epsilon$ transitions, c) Event determinism: Two edges with same source and same label have disjoint guards ($g_1$ $\cap$ $g_2$ = $\emptyset$), d) Time determinism: For every transition, the intersection of g with $I_q$ is at most a singleton.

\subsection{Expressiveness of $\epsilon$ transitions}
$\epsilon$ transitions add to the expressiveness of a timed automaton, i.e. the language recognized by a timed automaton with $\epsilon$ transitions may not be recognized by a timed automaton without the $\epsilon$ transitions. Take the automaton in Fig. 4. as an example. This automaton accepts timed words over \textit{a} such that every occurrence time is an integer and no two \textit{a} events occur at the same time. This language cannot be accepted by a timed automaton if $\epsilon$ switches are not allowed. If the largest constant in such timed automaton is \textit{c}, then a timed automaton without $\epsilon$ transitions cannot distinguish between the words (a, c+1) and (a, c+1.1). In this case, it can't distinguish between (a, 2) and (a, 2.1).

\begin{figure}
\begin{center}
\begin{tikzpicture}[scale=0.2]
\tikzstyle{every node}+=[inner sep=0pt]
\draw [black] (37.6,-32.3) circle (3);
\draw [black] (37.6,-32.3) circle (2.4);
\draw [black] (36.277,-29.62) arc (234:-54:2.25);
\draw (37.6,-25.05) node [above] {$a,\mbox{ }x=1,\mbox{ }x:=0$};
\fill [black] (38.92,-29.62) -- (39.8,-29.27) -- (38.99,-28.68);
\draw [black] (38.923,-34.98) arc (54:-234:2.25);
\draw (37.6,-39.55) node [below] {$\epsilon,\mbox{ }x=1,\mbox{ }x:=0$};
\fill [black] (36.28,-34.98) -- (35.4,-35.33) -- (36.21,-35.92);
\end{tikzpicture}
\caption{Expressiveness of $\epsilon$ transitions} \label{fig4}
\end{center}
\end{figure}
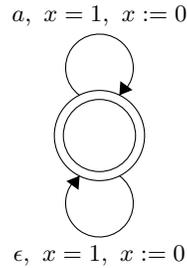

\subsection{Timed regular languages}
A timed language is called timed regular if it can be accepted by a timed automaton.

\subsection{Closure properties}
\begin{theorem}Timed regular languages are closed under the following operations.\end{theorem} Let us consider two timed regular languages, $L_1$, $L_2$.
\begin{enumerate}
   \item \textit{Union}: Union of two timed regular languages, ${L_1}\cup{L_2}$ is timed regular.
   \item \textit{Intersection}: Intersection of them, ${L_1}\cap{L_2}$ is timed regular.
   \item \textit{Projection}: Projection of a timed regular language is timed regular.
   \item \textit{Untime}: If L is timed regular, then untime(L) is $\omega$-regular [3].
\end{enumerate}

Closure under Union and Intersection are established by constructing product of the timed automata. Projection can be proved by labeling transitions with $\epsilon$. The proof for untime(L) being $\omega$-regular is established by region construction as shown in [1].

\subsection{Closure under complementation}
\begin{theorem}Timed regular languages are not closed under complementation [3].\end{theorem}
\begin{proof}
Let $\Sigma$ = \{a, b\}. L is a timed language consisting of timed words \textit{w}, containing an $\textit{a}$ event at some time $\textit{t}$ such that no event occurs at time $\textit{t + 1}$. \textit{L} is accepted by the timed automaton in Fig. 5. We will show that $\overline{L}$ is not timed regular. \textit{untime(L)} accepts \textit{(a+b)*a(a+b)*}.

\begin{figure}
\begin{center}
\begin{tikzpicture}[scale=0.2]
\tikzstyle{every node}+=[inner sep=0pt]
\draw [black] (28.5,-25.4) circle (3);
\draw (28.5,-25.4) node {$s_1$};
\draw [black] (54.5,-25.4) circle (3);
\draw (54.5,-25.4) node {$s_2$};
\draw [black] (54.5,-25.4) circle (2.4);
\draw [black] (31.5,-25.4) -- (51.5,-25.4);
\fill [black] (51.5,-25.4) -- (50.7,-24.9) -- (50.7,-25.9);
\draw (41.5,-25.9) node [below] {$a,\mbox{ }x:=0$};
\draw [black] (27.177,-22.72) arc (234:-54:2.25);
\draw (28.5,-18.15) node [above] {$a,b$};
\fill [black] (29.82,-22.72) -- (30.7,-22.37) -- (29.89,-21.78);
\draw [black] (53.177,-22.72) arc (234:-54:2.25);
\draw (54.5,-18.15) node [above] {$a,b\mbox{ }x\neq1$};
\fill [black] (55.82,-22.72) -- (56.7,-22.37) -- (55.89,-21.78);
\end{tikzpicture}
\caption{Timed Automaton for disproving closure on complementation} \label{fig6}
\end{center}
\end{figure}
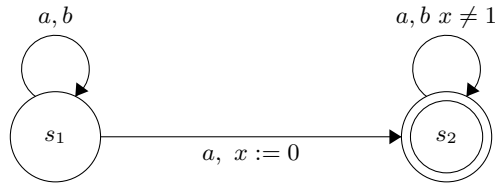

Consider timed regular language L', consisting of timed words \textit{w}, such that untimed word of \textit{w} is in \textit{a*b*}. All \textit{a} events happen before time 1 and no two \textit{a} events happen at same time. Observe the word $a^{m}b^{n}$ $\in$ untime($\overline{L}{\cap}L'$) iff m $\geq$ n. Timed regular languages are closed under intersection, \textit{untime(L)} of the timed regular language \textit{L} is regular, but since we also know that $a^{m}b^{n}$ m $\geq$ n is not regular, we proved that $\overline{L}$ is not timed regular. \end{proof}

\subsection{Emptiness}
Given a timed automaton over finite words, is the language accepted by it empty? This problem is known as the \textbf{emptiness} problem for timed automata. It can be analyzed by checking if there is a run of the automaton from an initial state to a final state. 
\begin{proof}
We have a problem in performing this check since the number of configurations is uncountably infinite due to real-valued clock times. The solution to this problem was proposed in [1] by constructing a finite region graph using region equivalence technique. The caveat here is that although there are infinite configurations, some clock valuations are equivalent and hence a timed automaton cannot distinguish between them. By grouping such configurations using region equivalence technique, we construct a finite region graph. The problem thus reduces to a reachability problem over this finite region graph, which is decidable. Hence it is proved that emptiness problem for timed automata is decidable.
\end{proof}

\subsection{Universality}
Given a timed automaton over finite words, is the language accepted by it the set consisting of every timed word? This problem is known as the \textbf{universality} problem for timed automata. This problem can be proved undecidable by reducing the halting problem over two-counter machines, which is known to be undecidable to the universality problem [3].

\subsection{Language inclusion}
Given two timed automata \textit{A} and \textit{B} over finite words, checking if L(A) $\subseteq$ L(B) is known as the \textbf{language inclusion} problem. This problem can be proved undecidable by reducing the halting problem over two-counter machines, which is known to be undecidable to the language inclusion problem [3].

\subsection{Decidability in special cases}
Let us consider the universality problem on a given timed automaton having at most one clock. In this special case, this problem has been proved to be decidable by Abdulla et al. in [4]. Given two timed automata A and B, such that the timed automaton B only has one clock, and the only constant appearing in the clock constraints of B is 0, the language inclusion problem on A and B, formulated as checking if L(A) $\subseteq$ L(B) has been proved to be decidable by Ouaknine et. al. in [5].

\subsection{Decidability of determinizability}
In [6], E. Asarin has posed an open question concerning timed automata as follows. \textit{Given a timed automaton A, is it possible to decide whether it is equivalent to a deterministic one?} This question remained open until O. Finkel has proved in [7] that this problem is undecidable.

\section{Hybrid Systems}
\subsection{Background}
Hybrid Systems are dynamical systems with interacting continuous-time dynamics and discrete-time dynamics. The evolution of the state of a continuous time system is described by an \textbf{ordinary differential equation} (ODE), $\dot{x} = Ax$, whereas the evolution of the state of a discrete-time system is described by a \textbf{difference equation}, $x_{k+1}$ = A$x_k$. The continuous time dynamics are referred to as flows and the discrete-time dynamics are referred to as jumps. The flows cause the system's state to make a smooth continuous transition and jumps cause the system to transition to a different set of flow equations by making a discrete jump to the new continuous-time dynamics. The hybrid behavior arrives in various contexts, a) Continuous systems with a phased operation, like the bouncing ball, biological cell growth to name a few, b) Continuous systems controlled by discrete logic, like control modes for complex systems, c) Coordinating processes, like multi-agent systems. When formally modeling the hybrid system behavior, the modeling paradigm we choose needs to have both continuous and discrete parts to it [8].

\subsection{Bouncing ball}
Let us consider an example of a ball of mass \textit{m} freely falling from height \textit{x} $\geq$ 0 that bounces after hitting the ground. The analysis of this system has two parts to it.
\subsubsection{Part-I Free Fall}
In free fall while x $\geq$ 0, the ball is under the influence of gravity. The equation of motion satisfies $\ddot{x}$ = -g, where \textit{x} is the position of the ball and \textit{g} is the gravitational constant. We reduce the order of the ODE by substituting $\dot{x}$ = v, where \textit{v} is the velocity of the ball. We have $\dot{v}$ = -g. This describes the continuous time dynamics of the system.
\subsubsection{Part-II Bouncing}
When the ball is at x = 0, and has a velocity downwards ($v<0$), the ball bounces on the ground. There is a loss of velocity due to deformation and friction. The velocity discretely jumps to a different value as per the equation v := -cv, where $c<1$. The velocity's magnitude is reduced and it's direction flips.
\subsubsection{Hybrid Automaton}
The hybrid automaton for bouncing ball is illustrated in Fig. 6. $x_1$ represents the vertical position and $x_2$ represents the velocity of the ball [9]. In Fig. 7 we illustrate the bouncing ball simulation in MATLAB, where the yellow legend corresponds to position of the ball, and blue legend corresponds to velocity of the ball.

\begin{figure}
\begin{minipage}{.5\textwidth}
\begin{center}
\begin{tikzpicture}[scale=0.175]
\tikzstyle{every node}+=[inner sep=0pt]
\draw [black] (38.2,-33.4) circle (5.9);
\draw (38.2,-31.4) node {$\dot{x_1}\mbox{ }=\mbox{ }x_2$};
\draw (38.2,-34.4) node {$\dot{x_2}\mbox{ }=\mbox{ }-g$};
\draw (38.2,-36.4) node {$m$};
\draw [black] (32.877,-30.72) arc (259:-12:4.25);
\draw (32.89,-17.96) node [above] {$x_1=0 \wedge x_2\leq0$};
\draw (32.89,-19.96) node [above] {impact};
\draw (32.89,-21.96) node [above] {${x_1}' = {x_2} \wedge {x_2}' = -c{x_2}$};
\draw (35.89,-40.96) node [below] {${x_1} \geq 0$};
\fill [black] (37.88,-27.66) -- (38.38,-26.45) -- (37.38,-26.46);
\end{tikzpicture}
\caption{Hybrid automaton} \label{fig6}
\end{center}
\end{minipage}
\begin{minipage}{.5\textwidth}
\begin{center}
\includegraphics[width=6.3cm,height=4.9cm,keepaspectratio]{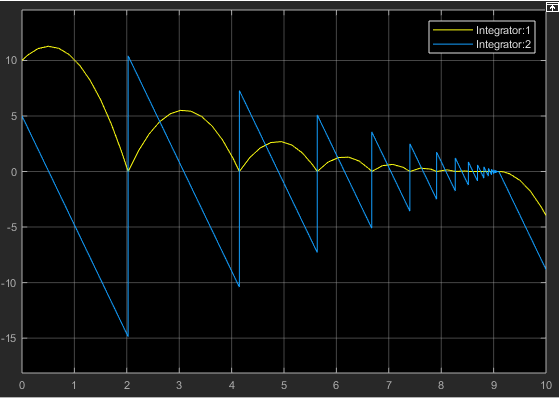}
\caption{Bouncing ball simulation} \label{fig7}
\end{center}
\end{minipage}
\end{figure}

\subsubsection{Formal defn.} A hybrid automaton H is defined as a tuple, H = (M, $\Sigma$, $M_0$, X, $\delta$, I, F, J). M is the finite set of control modes, also known as locations. $\Sigma$ is a finite set of actions. Actions are given as sets of differential equations. $M_0$ is a finite initial set of control modes, ${M_0}$ ${\subseteq}$ $M$. X is a finite set of real-valued variables. $\delta$ is a set of transitions. I is the mode-invariant function. F is the mode dependent flow function, characterizing flow at each mode. J is the jump function. The configuration of a hybrid automaton is $(m,v)$ where $m$ ${\in}$ $M$  is a mode/location and v is a variable valuation.

\subsubsection{Reachability}
Linear Hybrid Automata are a class of hybrid automata where its activities, invariants, and transition relations are defined by linear expressions. Given a linear hybrid automaton, is there a run reaching a particular state from the initial state? This problem is known as the \textbf{reachability} problem for linear hybrid automata. In [10], Alur et al. have reduced the halting problem of two-counter machines, which is known to be undecidable to the reachability problem for linear hybrid automata. Hence, it is proved to be undecidable.\\
Further, in [11], S.N. Krishna et al. have reduced the halting problem for two-counter machines to the reachability problem on recursive hybrid automata, thus showing that it is undecidable. However, bounded reachability is still decidable. There exist some incomplete algorithms for reachability. In [12], E. Abraham presents one such incomplete algorithm for linear hybrid automata based on fixed-point computation. Termination of this algorithm corresponds to finding the least fixed-point for the one-step forward reachability starting from the initial set.

\section{Cyber-Physical Systems}
\subsection{Background}
As defined in [13], a cyber-physical system consists of a collection of computing devices communicating with one another and interacting with the physical world through sensors and actuators in a feedback loop. There are quite a few applications of such systems everywhere, whether it is smart buildings or medical devices or even automobiles. One of the distinguishing characteristics of cyber-physical systems is that they are reactive in nature, meaning these systems interact with the environment in an ongoing manner through inputs and outputs, like a program for a cruise controller in a car just to name one. They are concurrent in nature, meaning as opposed to sequential computation, there are multiple threads known as processes executing concurrently. They have mixed discrete-continuous behavior such as the hybrid systems. They often have real-time system characteristics and are in the most safety-critical areas where errors could lead to catastrophic eventualities. We will study some crucial applications of cyber-physical systems and look at automata based modeling and analysis of such systems.

\subsection{Cardiac Pacemaker}
Implantable cardiac pacemakers are life-saving devices that help patients with heart arrhythmia conditions. Ironically, there are bugs even in such safety-critical devices. To guarantee the correct operation of such devices, we need a model of the heart that captures the physiological conditions of the heart and respond to pacemaker outputs, using which we can verify the safety properties of the pacemaker. The pacemaker functions autonomously according to the device algorithm that is implemented onto it, interacts with its environment through sensors, actuators and other devices. It is, therefore a perfect example of a cyber-physical system as defined in the previous section. Timed automata is an appropriate formalism for such a heart model since most timing behaviors of heart can be captured by timed automata. In [14] Z. Jiang et. al. have proposed a real-time heart model based on timed automata and capture such crucial timing properties of the heart to then verify the cardiac pacemaker.

\section{Uppaal case study}
\subsection{Background}
Uppaal is a model checker for real-time systems. We can conveniently model a real-time system as a network of timed automata running concurrently. Uppaal is used to verify system properties specified in CTL over the timed automata system model [15]. The CTL formulae can be specified over paths, and are classified broadly as \textit{safety}, \textit{liveness} and \textit{reachability} properties. Uppaal runs on a client-server architecture and is split into the backend model checking engine and frontend GUI communicating via TCP/IP protocol. It has a simulator where the user can run the system manually, and a verifier where the user may specify the properties and run the model checker on the system.

\subsection{Peterson's mutual exclusion}
The idea of mutual exclusion is that two processes which have critical sections cannot enter those sections at the same time. It is a technique for avoiding race conditions and keeping the result deterministic while running the processes concurrently. The Peterson's mutual exclusion algorithm is illustrated below.
\begin{algorithm}
\caption{Peterson's Mutex}\label{mutex}
\begin{minipage}{.5\textwidth}
\begin{algorithmic}[1]
\Procedure{Process 1}{}\\
req1 = 1;\\
turn = 2;\\
while(turn!=1 \&\& req2!=0);\\
\textit{//critical section}\\
job1();\\
req1 = 0;
\EndProcedure
\end{algorithmic}
\end{minipage}
\begin{minipage}{.5\textwidth}
\begin{algorithmic}[1]
\Procedure{Process 2}{}\\
req2 = 1;\\
turn = 1;\\
while(turn!=2 \&\& req1!=0);\\
\textit{//critical section}\\
job2();\\
req2 = 0;
\EndProcedure
\end{algorithmic}
\end{minipage}
\end{algorithm}

Fig. 8. illustrates the results of Uppaal model checker on the correctly implemented mutex algorithm as above. First property species that there exists a run where critical section is reachable as $E<> (P1.CS)$. Second property captures the mutual exclusion property, where on all paths both P1 and P2 cannot be in the critical section at the same time as A$[]$ not (P1.CS and P2.CS). We see that both the properties are satisfied in this case.

\begin{figure}
\begin{center}
\includegraphics[scale=0.35]{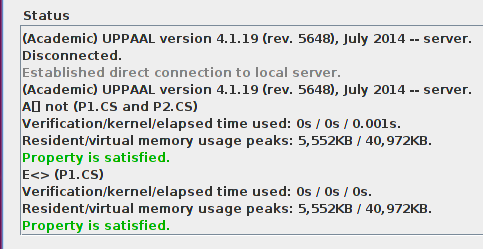}
\caption{Model checking mutex} \label{fig8}
\end{center}
\end{figure}

Fig. 9. illustrates incorrectly implemented mutex, where two processes are seen to be in the \textit{criticial section} at the same time in the simulator. Fig. 10 illustrates the results where Uppaal says that mutual exclusion property is not satisfied on the incorrectly implemented mutex. Mutual exclusion property has failed to satisfy the model.

\begin{figure}
\begin{minipage}{.5\textwidth}
\begin{center}
\includegraphics[scale=0.25]{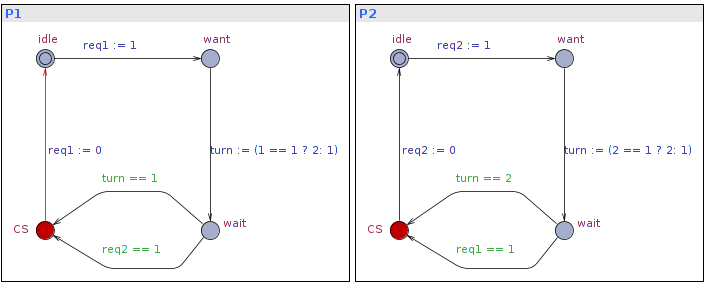}
\caption{Mutex protocol with bug} \label{fig9}
\end{center}
\end{minipage}
\begin{minipage}{.5\textwidth}
\begin{center}
\includegraphics[scale=0.31]{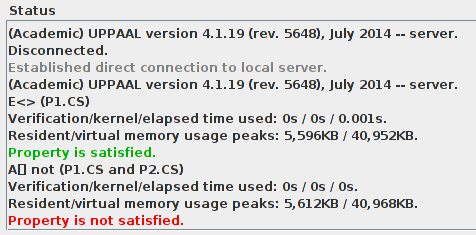}
\caption{Model checking mutex with bug} \label{fig10}
\end{center}
\end{minipage}
\end{figure}

\section{Conclusion}
The motive of this work is to study extensions to finite automata theory used to model systems of the 21st century that are becoming increasingly complex to verify, due to the confluence of sensors, actuators, real-time constraints, and their concurrent and reactive nature. The automata theory based modeling paradigms discussed above are helping us design computer-based models of such systems that can be used to analyze their behavior mathematically and ultimately build safety-critical systems that are more dependable and secure although complex.\\

As we use model checking to verify the properties of such systems, we also see that some seemingly simple problems like reachability don't have a decidable algorithm that works for any general system. Researchers in this area are working towards devising algorithms to model check these systems. These algorithms although while being incomplete algorithms, terminate in most cases that are very relevant to the system analysis. Future work also lies ahead in researching methods to synthesize correct-by-construct designs for controllers and other components involved in such real-time hybrid systems. Some open problems in hybrid systems are concerning non-linear hybrid systems. Scalability is a challenge for most of the modeling tools available to us currently. Furthermore, high-dimensional systems pose a unique challenge to reachability analysis in these tools. Recent works such as [16], show promise in this direction. Schupp et al. [17] present an analysis of the current challenges in verification of hybrid systems.

%
%
%
%

\end{document}